\documentclass{article}
\usepackage{amssymb, amsmath, amsthm}

\usepackage[english]{babel}
\usepackage[utf8]{inputenc}
\usepackage{textcomp}
\usepackage{hyperref}
\usepackage[T1]{fontenc}
\usepackage{comment}
\usepackage{todonotes}
\usepackage[capitalise]{cleveref}
\usepackage{fullpage} 
\usepackage{mathtools}
\usepackage{caption}
\usepackage{multirow}
\usepackage{wrapfig}

\title{No Repetition: Fast Streaming with Highly Concentrated Hashing}
\author{Anders Aamand\footnote{Basic Algorithms Research Copenhagen
    (BARC), University of Copenhagen.} \and Debarati Das$^*$ \and
  Evangelos Kipouridis$^*$ \and Jakob B.\ T.\ Knudsen$^*$ \and Peter
  M. R. Rasmussen$^*$ \and Mikkel Thorup$^*$}

\date{\today}

\newcommand{\abs}[1]{\left\vert #1\right\vert}

\newcommand{\eps}{\varepsilon}

\newcommand{\Z}{\mathbb{Z}}

\newcommand{\C}{\mathcal C}

\newcommand{\PR}[1]{\Pr\left[ #1 \right]}

\newcommand\drop[1]{}

\newcommand{\req}[1]{\eqref{#1}}

\newcommand{\cE}{{\mathcal E}}
\newcommand{\E}[1]{\mathbb E\left[#1\right]}

\newcommand{\Var}[1]{\mathrm{Var}\left[ #1 \right]}

\def\cC{\mathcal{C}}
\def\deltaerr{\varepsilon}

\makeatletter
\newtheorem*{rep@theorem}{\rep@title}
\newcommand{\newreptheorem}[2]{%
\newenvironment{rep#1}[1]{%
 \def\rep@title{#2 \ref{##1}}%
 \begin{rep@theorem}}%
 {\end{rep@theorem}}}
\makeatother

\newtheorem{definition}{Definition}
\newtheorem{theorem}{Theorem}
\newtheorem{lemma}[theorem]{Lemma}
\newreptheorem{lemma}{Lemma}

\begin{document}
\maketitle\begin{abstract} To get estimators that work within a
certain error bound with high probability, a common strategy is to
design one that works with constant probability, and then boost the
probability using independent repetitions.  Important examples of this
approach are small space algorithms for estimating the number of
distinct elements in a stream, or estimating the set similarity
between large sets.  Using standard strongly universal hashing to
process each element, we get a sketch based estimator where the
probability of a too large error is, say, 1/4. By performing $r$
independent repetitions and taking the median of the estimators, the
error probability falls exponentially in $r$. However, running $r$
independent experiments increases the processing time by a factor $r$.

Here we make the point that if we have a hash function with strong
concentration bounds, then we get the same high probability bounds
without any need for repetitions.  Instead of $r$ independent
sketches, we have a single sketch that is $r$ times bigger, so the
total space is the same. However, we only apply a single hash function, so
we save a factor $r$ in time, and the overall algorithms just get simpler.

Fast practical hash functions with strong concentration bounds were
recently proposed by Aamand {\em et al.} (to appear in {\em STOC
  2020}). Using their hashing schemes, the algorithms thus
become very fast and practical, suitable for online processing of high
volume data streams.
\end{abstract}

\newpage

\section{Introduction}
To get estimators that work within a certain error bound with high
probability, a common strategy is to design one that works with
constant probability, and then boost the probability using
independent repetitions.  A classic example of this approach is the
algorithm of Bar-Yossef {\em et al.}~\cite{BJKST02} to estimate the number
of distinct elements in a stream. Using standard strongly universal
hashing to process each element, we get an estimator where the
probability of a too large error is, say, 1/4. By performing $r$
independent repetitions and taking the median of the estimators, the
error probability falls exponentially in $r$. However, running $r$
independent experiments increases the processing time by a factor $r$.

Here we make the point that if we have a hash function with strong
concentration bounds, then we get the same high probability bounds
without any need for repetitions.  Instead of $r$ independent
sketches, we have a single sketch that is $\Theta(r)$ times bigger, so the
total space is essentially the same. However, we only apply a single hash
function, processing each element in constant time regardless of $r$,
and the overall algorithms just get simpler.

Fast practical hash functions with strong concentration bounds were
recently proposed by Aamand {\em et al.}~\cite{AKKRT20}.
Using their hashing schemes, we get a very fast implementation
of the above streaming algorithm, suitable for online processing of high
volume data streams.

To illustrate a streaming scenario where the constant in the
processing time is critical, consider the Internet.  Suppose we want
to process packets passing through a high-end Internet router. Each
application only gets very limited time to look at the packet before
it is forwarded.  If it is not done in time, the information is
lost. Since processors and routers use some of the same technology, we
never expect to have more than a few instructions available. Slowing
down the Internet is typically not an option.  The papers of
Krishnamurthy {\em et al.}~\cite{KSZC03} and Thorup and Zhang~\cite{thorup12kwise}
explain in more detail how high speed hashing is
necessary for their Internet traffic analysis.  Incidentally, the hash
function we use from \cite{AKKRT20} is a bit faster than the ones from
\cite{KSZC03,thorup12kwise}, which do not provide Chernoff-style
concentration bounds.

The idea is generic and can be applied to other algorithms. We
will also apply it to Broder's original min-hash
algorithm  \cite{Broder97onthe} to estimate set similarity, which
can now be implemented efficiently, giving the desired estimates with high probability.

\paragraph{Concentration}
Let us now be more specific about the algorithmic context. We have
a key universe $U$, e.g., 64-bit keys, and a random hash function $h$ mapping
$U$ uniformly into $R=(0,1]$.

For some input set $S$ and some fraction $p\in [0,1)$, we want to know
  the number $X$ of keys from $S$ that hash below $p$. Here $p$ could be
  an unknown function of $S$, but $p$ should be independent of
  the random hash function $h$. Then the mean $\mu$ is $\E X =  |S|p$.

If the hash function $h$ is fully random, we get the 
classic Chernoff bounds on $X$ (see, e.g, \cite{motwani95book}):

\begin{align}\label{eq:classic-chernoff+}
\PR{X\ge (1+\deltaerr)\mu}& \le \exp(-\deltaerr^2\mu/3)\textnormal{ for }0\leq \deltaerr\leq 1,\\
\label{eq:classic-chernoff-}
\PR{X\le (1-\deltaerr)\mu}& \le \exp(-\deltaerr^2\mu/2)\textnormal{ for }0\leq \deltaerr\leq 1.
\end{align}
Unfortunately, we cannot implement fully random hash functions as
it requires space as big as the universe.

To get something implementable in practice, Wegman and Carter
\cite{wegman81kwise} proposed strongly universal hashing. The random
hash function $h:U\rightarrow R$ is {\em strongly universal} if for
any given distinct keys $x,y\in U$, $(h(x),h(y))$ is uniform in $R^2$.
The standard implementation of a strongly universal hash function into
$[0,1)$ is to pick large prime $\wp$ and two uniformly random numbers $a,b\in
\Z_\wp$. Then $h_{a,b}(x)=((ax+b)\bmod\wp)/\wp$ is strongly universal
from $U\subseteq \Z_\wp$ to $R=\{i/\wp|i\in Z_\wp\}\subset[0,1)$. Obviously
it is not uniform in $[0,1)$, but for any $p\in [0,1)$, we have 
$\PR{h(x)<p}\approx p$ with equality if $p\in R$. Below we ignore
this deviation from uniformity in $[0,1)$.
      
Assuming we have a strongly universal hash function $h:U\rightarrow
[0,1)$, we again let $X$ be the number of elements from $S$ that hash
below $p$.  Then $\mu=\E{X}=|S|p$ and because the hash values are
2-independent, we have $\Var{X}\leq \E{X}=\mu$. Therefore, by
Chebyshev's inequality,
\[\PR{|X-\mu|\geq \eps\mu}<1/(\eps^2\mu).\]
As $\eps^2\mu$ gets large, we see that the concentration we get with
strongly universal hashing is much weaker than the Chernoff bounds with 
fully random hashing. However, Chebyshev is fine if we just
aim at a constant error probability like $1/4$, and then we can use
the median over independent repetitions to reduce the error probability.

In this paper we discuss benefits of having hash functions with strong concentration
akin to that of fully random hashing:
\begin{definition}\label{def:strong-concentration}
A hash function $h:U\rightarrow[0,1)$ is {\em strongly concentrated
    with added error probability $\cE$\/} if for any set
  $S\subseteq U$ and $p\in [0,1)$, if $X$ is the number of elements from
    $S$ hashing below $p$, $\mu=p|S|$ and $\eps\leq 1$, then
\[\PR{|X-\mu|\geq \eps\mu} =2\exp(-\Omega(\eps^2\mu))+\cE.\]
If $\cE=0$, we simply say that $h$ is {\em strongly concentrated}.
\end{definition}
Another way of viewing the added error probability $\cE$ is as follows. We have strong concentration as long as we do not aim for error probabilities
below $\cE$, so if $\cE$ is sufficiently low, we can simply ignore it.

What makes this definition interesting in practice is that Aamand {\em et
al.}~\cite{AKKRT20} recently presented a fast practical small
constant time hash function that for $U=[u]=\{0,\ldots,u-1\}$ is
strongly concentrated with added error probability $u^{-\gamma}$ for any
constant $\gamma$. This term is so small that we can ignore it
in all our applications.  The speed is obtained using certain character
tables in cache that we will discuss later.

Next we consider our two streaming applications, distinct elements and
set-similarity, showing how strongly concentrated hashing eliminates
the need for time consuming independent repetitions. We stress that in
streaming algorithms on high volume data streams, speed is of critical
importance. If the data is not processed quickly, the information is
lost.

Distinct elements is the simplest case, and here we will also discuss
the ramifications of employing the strongly concentrated hashing of
Aamand {\em et al.}~\cite{AKKRT20} as well as possible alternatives.

\section{Counting distinct elements in a data stream}
We consider a sequence of keys $x_1,\ldots,x_s \in [u]$ where each
element may appear multiple times. Using only little space, we wish to
estimate the number $n$ of distinct keys. We are given parameters
$\eps$ and $\delta$, and the goal is to create an estimator, $\hat n$, such that $(1-\eps) n \leq \hat n \leq  (1+\eps) n$ with probability at least $1-\delta$.

Following the classic approach of Bar-Yossef {\em et al.}~\cite{BJKST02}, we
use a strongly universal hash function $h:U\to (0,1]$. For simplicity,
  we assume $h$ to be collision free over $U$.

  For some $k>1$, we maintain the $k$ smallest distinct hash
  values of the stream. We assume for simplicity that $k\leq n$. The space required is thus $O(k)$, so we want $k$ to be small.  Let $x_{(k)}$ be the key having the $k$'th smallest hash value under $h$ and let $h_{(k)}=h(x_{(k)})$. As in~\cite{BJKST02}, we use $\hat n=k/h_{(k)}$ as an estimator for
  $n$ (we note that \cite{BJKST02} suggests several other estimators,
  but the points we will make below apply to all of them).

The point in using a hash function $h$ is that all occurrences of a given
key $x$ in the stream get the same hash value, so if
$S$ is the set of distinct keys, $h_{(k)}$ is
just the $k$ smallest hash value from $S$. In particular, $\hat n$ depends only on $S$, not on the frequencies of the elements of the stream.
Assuming no
collisions, we will often identify the elements with
the hash values, so $x_i$ is smaller than $x_j$ if
$h(x_i)\leq h(x_j)$.

We would like $1/h_{(k)}$ to be concentrated
around $n/k$. For any probability $p\in[0,1]$, let
$X^{<p}$ denote the number of elements from $S$ that hash below $p$.
Let $p_-=k/((1+\eps)n)$ and $p_+=k/((1-\eps)n)$. Note that both $p_-$ and
$p_+$ are independent of the random hash function $h$. Now 
\begin{align*}
1/h_{(k)}&\leq (1-\eps)n/k\iff X^{<p_+}< k=(1-\eps)\E{X^{<p_+}}\\
1/h_{(k)}&> (1+\eps)n/k\iff X^{<p_-}\geq k=(1+\eps)\E{X^{<p_-}},
\end{align*}
and these observations form a good starting point for applying probabilistic tail bounds as we now describe.
\subsection{Strong universality and independent repetitions}
Since $h$ is strongly universal, the hash values of any two keys are
independent, so for any $p$, we have $\Var{X^{<p}}\leq \E{X^{<p}}$, and so
by Chebyshev's inequality,
 \begin{align*}
\PR{1/h_{(k)}\leq (1-\eps)n/k}&<(1-\eps)/(k\eps^2)\\
\PR{1/h_{(k)}>(1+\eps)n/k}&\leq(1+\eps)/(k\eps^2).
\end{align*}
Assuming $\eps\leq 1$, we thus get that
\[\PR{|\hat n -n|> \eps n}=\PR{\left|1/h_{(k)}-n/k\right| > \eps n/k}\leq 2/(k\eps^2).\]
To get the desired error probability $\delta$, we could now set
$k=2/(\delta\eps^2)$, but if $\delta$ is small, e.g. $\delta=1/u$, $k$ becomes way too large. As in~\cite{BJKST02} we instead start by aiming for a constant error probability, $\delta_0$, say $\delta_0=1/4$. For this value of $\delta_0$, it suffices to set $k_0=8/\eps^2$. We now run $r$ (to be
determined) independent experiments with this value of $k_0$, obtaining independent estimators for $n$, $\hat n_1,\ldots,\hat n_r$. Finally, as our final estimator, $\hat n$, we return the median of $\hat n_1,\dots,\hat n_r$. Now for each $1\leq i \leq r$, $\Pr[|\hat n_i-n|> \eps n]\leq 1/4$ and these events are independent. If $|\hat n-n|\geq \eps n$, then $|\hat n_i-n|\geq \eps n$ for at least half of the $1\leq i \leq r$. By the standard Chernoff bound~\req{eq:classic-chernoff+}, this probability can be bounded by 
\[\PR{|\hat n -n|> \eps n}
\leq \exp(-(r/4)/3)=\exp(-r/12).\] 
Setting $r=12\ln(1/\delta)$, we get
the desired error probability $1/\delta$. The total number of hash
values stored is
$k_0r=(8/\eps^2)(12\ln(\delta))=96\ln(1/\delta)/\eps^2$.

\subsection{A better world with fully random hashing}
Suppose instead that $h:[u]\to (0,1]$ is a fully random hash function.
In this case, the standard Chernoff bounds \req{eq:classic-chernoff+} and
\req{eq:classic-chernoff-} with $\eps\leq 1$ yield
\begin{align*}
\PR{1/h_{(k)}\leq (1-\eps)n/k}&<\exp(-(k/(1-\eps))\eps^2/2)\\
\PR{1/h_{(k)}>(1+\eps)n/k}&\leq\exp(-(k/(1+\eps))\eps^2/3).
\end{align*}
Hence
\begin{equation}\label{eq:distinct-Chernoff}
\PR{|\hat n -n|> \eps n}=\PR{|1/h_{(k)}-n/k|\geq \eps n/k}
\leq 2\exp(-k\eps^2/6).
\end{equation}
Thus, to get error probability $\delta$, we just use $k=6\ln(2/\delta)/\eps^2$.
There are several reasons why this is much better than the above approach
using 2-independence and independent repetitions.
\begin{itemize}
\item It avoids the independent repetitions, so instead of
applying $r=\Theta(\log(1/\delta))$ hash functions to each key we just need one. We thus save a factor of $\Theta(\log(1/\delta))$ in speed. 
\item Overall we store fewer hash values: 
$k=6\ln(2/\delta)/\eps^2$
instead of $96\ln(1/\delta)/\eps^2$.
\item With independent repetitions, we are tuning the algorithm depending on $\eps$ and
$\delta$, whereas with a fully-random hash function, we get the
concentration from \req{eq:distinct-Chernoff} for every $\eps\leq 1$.
\end{itemize}
The only caveat is that fully-random hash functions cannot be implemented.
\subsection{Using hashing with strong concentration bounds}\label{sec:distinct-strong}
We now discuss the effect of relaxing the abstract full-random hashing
to hashing with strong concentration bounds and added error probability $\cE$. Then for $\eps\leq 1$,
\begin{align*}
\PR{1/h_{(k)}\leq (1-\eps)n/k}&=2\exp(-\Omega(k/(1-\eps))\eps^2)+\cE\\
\PR{1/h_{(k)}>(1+\eps)n/k}&=2\exp(-\Omega(k/(1+\eps))\eps^2)+\cE.
\end{align*}
so 
\begin{equation}\label{eq:distinct-strong}
\PR{|\hat n-n| \geq \eps n}=\PR{|1/h_{(k)}-n/k| \geq \eps n/k}
\leq 2\exp(-\Omega(k\eps^2))+O(\cE).
\end{equation}
To obtain the error probability $\delta=\omega(\cE)$, 
we again need to store $k=O(\log(1/\delta)/\eps^2)$ hash values.
Within a constant factor this means that we use the same total number
using 2-independence and independent repetitions, and we still
retain the following advantages from the fully random case.
\begin{itemize}
\item With no independent repetitions we avoid
applying $r=\Theta(\log(1/\delta))$ hash functions to each key,
so we basically save a factor $\Theta(\log(1/\delta))$ in speed. 
\item With independent repetitions, we only address a given $\eps\leq 1$ and
$\delta$, while with a fully-random hash function we get the
concentration from \req{eq:distinct-Chernoff} for every $\eps\leq 1$.
\end{itemize}

\subsection{Implementation and alternatives}
We briefly discuss how to maintain the $k$ smallest elements/hash values.  The
most obvious method is using a priority queue, but this takes
$O(\log k)$ time per element, dominating the cost of evaluating the hash
function. However, we can get down to constant time per element if we
have a buffer for $k$. When the buffer gets full, we find the median in linear time with (randomized) selection and discard the bigger elements. This is standard to de-amortize if needed. 

A different, and more efficient, sketch from~\cite{BJKST02} identifies
the smallest $b$ such that the number $X^{<1/2^b}$ of keys hashing
below $1/2^b$ is at most $k$. For the online processing of the stream,
this means that we increment $b$ whenever $X^{<1/2^b}>k$. At the end,
we return $2^b X^{<1/2^b}$.  The analysis of this alternative sketch
is similar to the one above, and we get the same advantage of avoiding
independent repetitions using strongly concentrated hashing, that is,
for error probability $\delta$, in~\cite{BJKST02}, they run $O(\log (1/\delta))$
independent experiments with independent hash functions,
each storing up to $k=O(1/\eps^2)$ hash values, whereas we run only a
single experiment with a single strongly concentrated hash function
storing $k=O(\log (1/\delta)/\eps^2)$ hash values. The total number of
hash values stored is the same, but asymptotically,
we save a factor $\log (1/\delta)$ in time.
  
\paragraph{Other alternatives}

Estimating the number of distinct elements in a stream began with the
work of Flajolet and Martin \cite{Flajolet85counting} and has
continued with a long line of research \cite{alon96frequency, BJKST02,
  BKS02, BHRSG07, BC09, cohen97minwise, DF03, EVF06,
  Flajolet85counting,FlajoletFG07, Gibbons01, GibbonsT01, IndykW03,
  Woodruff04}. In particular, there has been a lot of focus on
minimizing the sketch size.  Theoretically speaking, the problem
finally found an asymptotically optimal, both in time and in space,
solution by Kane, Nelson and Woodruff \cite{KNW10:stream}, assuming we
only need $\frac{2}{3}$ probability of success. The optimal space,
including that of the hash function, is $O(\eps^{-2}+\log n)$ bits,
improving the $O(\eps^{-2} \cdot \log n)$ bits needed by Bar-Yossef {\em et
al.}~\cite{BJKST02} to store $O(\eps^{-2})$ hash values. Both
\cite{BJKST02} and \cite{KNW10:stream}, suggest using $O(\log
(1/\delta))$ independent repetitions to reduce the error probability to
$1/\delta$, but then both time and space blow up by a factor $O(\log
(1/\delta))$.

Recently Blasiok \cite{Bla18:stream} found a space optimal algorithm
for the case of small error probability $1/\delta$. In this
case, the bound from \cite{KNW10:stream} with independent
repetitions was $O(\log (1/\delta)(\eps^{-2}+\log n))$
which he reduces to $O(\log (1/\delta)\eps^{-2}+\log n)$, again including
the space for hash functions. He no longer has $O(\log (1/\delta))$ hash
functions, but this only helps his space, not his processing time, which
he states as polynomial in $\log (1/\delta)$ and $\log n$.

The above space optimal algorithms \cite{Bla18:stream,KNW10:stream}
are very interesting, but fairly complicated, seemingly involving some
quite large constants. However, here our focus is to get a fast practical
algorithm to handle a high volume data stream online, not worrying as
much about space. Assuming fast strongly concentrated hashing, it is
then much better to use our implementation of the simple algorithm of
Bar-Yossef {\em et al.}~\cite{BJKST02} using $k=O(\eps^{-2}\log
(1/\delta))$.

\subsection{Implementing Hashing with Strong Concentration}
As mentioned earlier, Aamand {\em et al.}~\cite{AKKRT20} recently presented
a fast practical small constant time hash function,
Tabulation-1Permutation, that for $U=[u]=\{0,\ldots,u-1\}$ is strongly
concentrated with additive error $u^{-\gamma}$ for any constant
$\gamma$. The scheme obtains its power and speed using certain
character tables in cache.

More specifically, we view keys as consisting of a small number $c$ of
characters from some alphabet $\Sigma$, that is, $U=\Sigma^c$. For
$64$-bit keys, this could be $c=8$ characters of $8$ bits each. Let's say
that hash values are also from $U$, but viewed as bit strings representing
fractions in $[0,1)$.

Tabulation-1Permutation needs $c+1$ character tables mapping
characters to hash values. To compute the hash value of a key, we need
to look up $c+1$ characters in these tables. In addition we need
$O(c)$ fast AC$^0$ operations to extract the characters and xor the
hash values. The character tables can be populated with an $O(\log n)$
independent pseudo-random number generator, needing a random seed of
$O((\log n)(\log u))$ bits.

\paragraph{Computer dependent versus problem dependent view of resources for hashing}
We view the resources used for Tabulation-1Permutation as computer
dependent rather than problem dependent. When you buy a new computer
you can decide how much cache you want to allocate for your hash functions.
In the experiments performed in~\cite{AKKRT20}, using 8-bit characters
and $c=8$ for 64-bit keys was very efficient. On two computers, it was
found that tabulation-1permutation was less than 3 times slower than
the fastest known strongly universal hashing scheme; namely Dietzfelbinger's
\cite{dietzfel96universal} which does just one multiplication and one shift.
Also, Tabulation-1Permutation was more than 50 times faster than the
fastest known highly independent hashing scheme; namely Thorup's
\cite{Tho13:simple-simple} double tabulation scheme which, in theory
also works in constant time.

In total, the space used by all the character tables is $9\times 2^8\times 64$
bits which is less than 20 KB, which indeed fits in very fast cache. We note
that when we have first populated the tables with hash values, they are not
overwritten. This means that the cache does not get dirty, that is
different computer cores can access the tables and not worry about
consistency.

This is different than the work space used to maintain the sketch
of the number of distinct keys represented via $k=O(\eps^{-2}\log (1/\delta))$ 
hash values, but let's compare anyway with real numbers. Even with
a fully random hash function with perfect Chernoff bounds, we needed
$k=6\ln(2/\delta)/\eps^2$,
so with, say, $\delta=1/2^{30}$ and $\eps=1\%$, we
get $k> 2^{20}$, which is much more than the $9\times 2^8$ hash values
stored in the character tables for the hash functions. Of course, we would be
happy with a much smaller $k$ so that everything is small and fits in
fast cache.

We note that any $k>|\Sigma|=2^8$ rules out the concentration
of previous tabulation schemes such a simple tabulation \cite{patrascu12charhash}
and twisted tabulation \cite{PT13:twist}. The reader is referred
to~\cite{AKKRT20} for a thorough discussion of the alternatives.

Finally, we relate our strong concentration from Definition
\ref{def:strong-concentration} to the exact concentration result
from~\cite{AKKRT20}:
\begin{theorem}\label{thm:tab-1perm}
     Let $h\colon [u]\to [r]$ be a tabulation-1permutation hash
     function with $[u]=\Sigma^c$ and $[r]=\Sigma^d$, $c,d=O(1)$. Consider
a key/ball set $S\subseteq [u]$ of size $n=\abs{S}$ where each ball $x\in S$ 
is assigned a weight $w_x \in [0,1]$.  Choose arbitrary hash values $y_1, y_2\in [r]$ with
$y_1\leq y_2$.  Define $X=\sum_{x\in
S}w_x\cdot [y_1\leq h(x)< y_2]$ to be the total weight of balls hashing to
the
interval $[y_1, y_2)$. Write $\mu = \E{X}$ and $\sigma^2=\Var{X}$. Then for any constant
$\gamma$ and every $t>0$, 
\begin{equation}\label{eq:concentration}
\Pr[|X-\mu|\geq  t] \le 2\exp(-\Omega(\sigma^2 \; \cC(t/\sigma^2)))+1/u^\gamma.
\end{equation}
Here $\mathcal{C}: (-1,\infty) \to [0,\infty)$ is given by
  $\cC(x)=(x+1)\ln (x+1)-x$, so $\exp(-\C(x))=
  \frac{e^x}{(1+x)^{(1+x)}}$. The above also holds if we condition the
  random hash function $h$ on a distinguished query key $q$ having a
  specific hash value.
\end{theorem}
The above statement is far more general than what we need. All our
weights are unit weights. We fix $r=u$ and $y_1=0$. Viewing hash
values as fractions in $[0,1)$, the random variable $X$ is the number
  of items hashing below $p=y_2/u$.  Also, since $\Var{X}\leq \E{X}$,
  \req{eq:concentration} implies the same statement with $\mu$ instead of
  $\sigma^2$. Moreover, our $\eps\leq 1$ corresponds to $t=\eps\mu\leq
  \mu$, and then we get
\[\Pr[|X-\mu|\geq  \eps\mu] \le 2\exp(-\Omega(\mu\; \cC(\eps)))+1/u^\gamma \leq
2\exp(-\Omega(\mu\eps^2))+1/u^\gamma.\]
which is exactly as in our Definition \ref{def:strong-concentration}. Only
remaining difference is that Definition \ref{def:strong-concentration}
should work for {\em any} $p\in [0,1)$ while the bound we get only
works for $p$ that are multiples of $1/u$. However, this suffices
by the following general lemma:
\begin{lemma} Suppose we have a hash function $h:[u]\rightarrow[0,1)$
    such that for any set $S\subseteq U$ and for any $p\in [0,1)$ that is a
    multiple of $1/u$, for the number $X^{<p}$ of elements from $S$
    that hash below $p$, with $\mu_{p}=p|S|$ and $\eps\leq 1$, it holds that
\[\PR{|X^{<p}-\mu_p|\geq \eps\mu_p} \le 2\exp(-\Omega(\eps^2\mu_p))+O(\cE).\]
Then the same statement holds for all $p\in[0,1)$
\end{lemma}
\begin{proof} First we note that the statement is trivially true if $\eps^2\mu_p=O(1)$, so we can assume $\eps^2\mu_p=\omega(1)$. Since $\eps\leq 1$,
we also have $\mu_p=\omega(1)$.

We are given an arbitrary $p\in [0,1)$. Let  $p_{+}=i/u$ be the nearest higher multiple    of $1/u$. Since $|S|\leq u$ and $\mu_p=p|S|$ we have $i\geq\mu_p$, implying $i=\omega(1)$. 
We also let $p_{-}=(i-1)/u$.

It is now clear that since $p_{-}<p\le p_{+}$, it holds that $X^{<p_{-}}\le X^{<p}\le X^{<p_{+}}$. We first show that
\[X^{<p}\le (1-\eps)\mu_p \implies X^{<p_{-}}\le (1-\eps/2)\mu_{p_{-}}.\]

Indeed, $X^{<p}\le (1-\eps)\mu_p$ implies $X^{<p_{-}}\le (1-\eps)p|S| \le (1-\eps)(p_{-}+1/u)|S| = \mu_{p_{-}} - \eps \mu_{p_{-}} + (1-\eps)|S|/u$.

But $|S|\le u$ and $(1-\eps)<1$, so $X^{<p_{-}} \le \mu_{p_{-}} - \eps \mu_{p_{-}} + 1 \le (1-\eps/2)\mu_{p_{-}}$. The last follows from the fact that $(\eps/2) \mu_{p_{-}} \ge (\eps/2) \mu_p - (\eps/2) |S|/u \ge (\eps^2 /2) \mu_p - 1$, but $\eps^2 \mu_p = \omega(1)$ and so $(\eps/2) \mu_{p_{-}} = \omega(1)$.

The exact same reasoning gives 

\[X^{<p}\ge (1+\eps)\mu_p \implies X^{<p_{+}}\ge (1+\eps/2)\mu_{p_{+}}.\]

But then 
\[\PR{|X^{<p}-\mu_p|\geq \eps\mu_p} = 
\PR{X^{<p}\le (1-\eps)\mu_p} + \PR{X^{<p}\ge (1+\eps)\mu_p} \le \]
\[ \PR{X^{<p_{-}}\le (1-\eps/2)\mu_{p_{-}}} + \PR{X^{<p_{+}}\ge (1+\eps/2)\mu_{p_{+}}} \le \]
\[ \PR{|X^{<p_{-}}-\mu_{p_{-}}|\ge (\eps/2) \mu_{p_{-}}} + \PR{|X^{<p_{+}}-\mu_{p_{+}}| 	\ge (\eps/2) \mu_{p_{+}}} \le \]

Notice that $\mu_p-1 \le \mu_{p_{-}} \le \mu_{p_{+}}$, and $p_{-}$ and $p_{+}$ are
multiples of $1/u$, so we can use the bounds of the statement. Thus 
$\PR{|X^{<p}-\mu_p|\geq \eps\mu_p}$ is upper bounded by
\[4\exp(-\Omega((\eps/2)^2(\mu_{p}-1)))+O(\cE) = 2\exp(-\Omega(\eps^2\mu_{p}))+O(\cE)\]
\end{proof}
We note that~\cite{AKKRT20} also presents a slightly slower scheme, Tabulation-Permutation, which offers far more general concentration bounds than
those for Tabulation-1Permutation in Theorem \ref{thm:tab-1perm}. However,
Tabulation-1Permutation is faster and sufficient for the strong concentration
needed for our streaming applications.

\section{Set similarity} 
\newcommand\MIN{\textnormal{MIN}}
We now consider Broder's~\cite{Broder97onthe}
original algorithm for set similarity. As above,
it uses a hash function $h:[u]\to[0,1]$ which we
assume to be collision free. The bottom-$k$ sample $\MIN_k(S)$ of a set 
$S\subseteq [u]$ consists of the $k$ elements with the smallest
hash values. If $h$ is fully random then $\MIN_k(S)$ is a uniformly
random subset of $k$ distinct elements from $\MIN_k(S)$. We assume
here that $k\leq n=|S|$.
With $\MIN_k(S)$, we can
estimate the frequency $f=|T|/|S|$ of any subset $T\subseteq S$ as 
$|\MIN_k(S)\cap T|/k$. 

Broder's main application is the estimation of the Jaccard
similarity $f=|A\cap B|/|A\cup B|$ between sets $A$ and $B$. Given the
bottom-$k$ samples from $A$ and $B$, we may construct the bottom-$k$ sample of
their union as $\MIN_k(A\cup B)=\MIN_k(\MIN_k(A)\cup \MIN_k(B))$, and then the
similarity is estimated as $|\MIN_k(A\cup B)\cap \MIN_k(A)\cap \MIN_k(B)|/k$.

We note again the crucial importance of having a common hash
function $h$. In a distributed setting, samples $\MIN_k(A)$ and $\MIN_k(B)$
can be generated by different entities. As long as
they agree on $h$, they only need to communicate the samples to estimate
the Jaccard similarity of $A$ and $B$. As noted before, for Tabulation-1Permutation $h$ can be shared by exchanging a random seed of $O((\log{n})(\log{u}))$ bits.

For the hash function $h$, Broder \cite{Broder97onthe} first considers
fully random hashing. Then $\MIN_k(S)$ is a fully random sample of
$k$ distinct elements from $S$, which is very well understood.

Broder also sketches some alternatives with realistic hash functions,
but Thorup \cite{thorup13bottomk} showed that even if we
just use 2-independence, we get the same expected error as with fully
random hashing, but here we want strong concentration. Our analysis follows the simple
union-bound approach from \cite{thorup13bottomk}.


For the analysis, it is simpler to study the case where
we are sampling from a set $S$ and want to estimate the
frequency $f=|T|/|S|$ of a subset $T\subseteq S$. Let
$h_{(k)}$ be the $k$th smallest hash value from $S$ as
in the above algorithm for estimating distinct elements.
For any $p$ let $Y^{\leq p}$ be the number of elements
from $T$ with hash value at most $p$. Then
$|T\cap \MIN_k(S)|=Y^{\leq h_{(k)}}$ which is our estimator for $fk$.

\begin{theorem}
For $\eps\leq 1$, if $h$ is strongly concentrated with added error probability $\cE$, then 
\begin{equation}\label{eq:sim-small-eps}
\PR{|Y^{\leq h_{(k)}}-fk|>\eps fk}=2\exp(-\Omega(fk\eps^2))+O(\cE).
\end{equation}
\end{theorem}
\begin{proof}
Let $n=|S|$. We already saw in \req{eq:distinct-strong} that for any $\eps_S\leq 1$, 
\(P_S=\PR{|1/h_{(k)}-n/k|\geq \eps_S n/k}\leq 2\exp(-\Omega(k\eps_S^2))+O(\cE).
\) 
Thus, with $p_-=k/((1+\eps_S)n)$ and $p_+=k/((1-\eps_S)n)$,
we have $h_{(k)}\in [p_-,\,p_+]$ with probability $1-P_S$, and in
that case, $Y^{\leq p_-}\leq Y^{\leq h_{(k)}}\leq Y^{\leq p_+}$. 

Let $\mu^-=\E{Y^{\leq p_-}}=fk/(1+\eps_S)\geq fk/2$.  By strong concentration,
for any $\eps_T\leq 1$, we get that 
\[P^-_{T}=\PR{Y^{\leq p_-}\leq (1-\eps_T)\mu_-}\leq 2\exp(-\Omega(\mu_-\eps_T^2))+\cE
= 2\exp(-\Omega(fk\eps_T^2))+\cE.\]
Thus
\[\PR{Y^{\leq h_{(k)}}\leq \frac{1-\eps_T}{1+\eps_S}fk}\leq P^-_T+P_S.\]
Likewise, with $\mu^+=\E{Y^{\leq p_+}}=fk/(1-\eps_S)$, for any $\eps_T$,
we get that
\begin{align*}
P^+_T=\PR{Y^{\leq p_+}\geq (1+\eps_T)\mu_+}
\leq 2\exp(-\Omega(\mu_+\eps_T^2))+\cE
= 2\exp(-\Omega(fk\eps_T^2))+\cE,
\end{align*}
and
\begin{align*}
\PR{Y^{\leq h_{(k)}}\geq \frac{1+\eps_T}{1-\eps_S}fk}\leq P^+_T+P_S.
\end{align*}
To prove the theorem for $\eps\leq 1$, we set $\eps_S=\eps_T=\eps/3$. Then
$\frac{1+\eps_T}{1-\eps_S}\leq 1+\eps$ and 
$\frac{1-\eps_T}{1+\eps_S}\geq 1-\eps$. Therefore
\[\PR{|Y^{\leq h_{(k)}}-fk|\geq \eps fk}\leq P^-_T+P^+_T+2P_S\leq
8\exp(-\Omega(fk\eps_T^2))+O(\cE)= 2\exp(-\Omega(fk\eps_T^2))+O(\cE).\]
This completes the proof
of \req{eq:sim-small-eps}.
\end{proof}
As for the problem of counting distinct elements in a stream, in the online setting we may again modify the algorithm above to obtain a more efficient sketch. Assuming that the elements from $S$ appear in a stream, we again identify the smallest $b$ such that the number of keys from $S$ hashing below $1/2^b$, $X^{\leq 1/2^b}$, is at most $k$. We increment $b$ by one whenever $X^{\leq 1/2^b}>k$ and in the end we return $Y^{\leq 1/2^b}/X^{\leq 1/2^b}$ as an estimator for $f$. The analysis of this modified algorithm is similar to the analysis provided above.

\clearpage
\section*{Acknowledgements}
\begin{wrapfigure}{l}{0.055\textwidth}
\vspace{-4 mm}
\includegraphics[height=0.055\textwidth, width=0.055\textwidth]{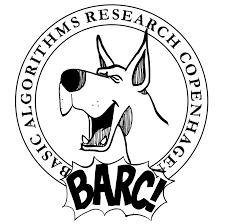} 
\includegraphics[height=0.03\textwidth, width=0.055\textwidth]{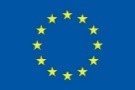} 
\end{wrapfigure}
Research of all authors partly 
supported by Thorup's Investigator Grant 16582, Basic Algorithms Research Copenhagen (BARC), from the VILLUM Foundation.
Evangelos Kipouridis has also received funding
from the European Union's Horizon 2020 research and innovation
program under the Marie Skłodowska-Curie grant agreement No 801199.

	\bibliographystyle{acm}
	\bibliography{bib}
\end{document}